\documentclass[10pt]{article}

\usepackage{amsfonts,amsmath,amssymb,amsthm,bm,graphicx,tikz}
\usepackage[authoryear,round]{natbib}
\usepackage[hyphenbreaks]{breakurl}
\usepackage[hyphens]{url}
\usepackage{geometry}

\DeclareMathOperator{\sgn}{sgn}
\DeclareMathOperator{\sat}{sat}

\newtheorem{theorem}{Theorem}

\theoremstyle{remark}
\newtheorem{remark}{Remark}
\newtheorem{assumption}{Assumption}
\title{An intelligent controller for underactuated mechanical systems}

\author{
Josiane Maria de Macedo Fernande, Marcelo Costa Tanaka,\\
Wallace Moreira Bessa, Edwin Kreuzer
}

\date{}

\begin{document}

\maketitle

\abstract{This paper presents an intelligent controller for uncertain underactuated nonlinear systems. The adopted approach 
is based on sliding mode control and enhanced by an artificial neural network to cope with modeling inaccuracies and external 
disturbances that can arise. The sliding surfaces are defined as a linear combination of both actuated and unactuated variables. 
A radial basis function is added to compensate the performance drop when, in order to avoid the chattering phenomenon, the sign
function is substituted by a saturation function in the conventional sliding mode controller. An application of the proposed 
scheme is introduced for an inverted pendulum, in order to illustrate the controller design method, and numerical results are 
presented to demonstrate the improved performance of the resulting intelligent controller.}

\section*{INTRODUCTION}

A mechanical system could be defined as underactuated if it has more degrees of freedom to be controlled than 
independent control inputs/actuators. Underactuated mechanical systems (UMS) play an essential role in several 
branches of industrial activity and their application scope ranges from robotic manipulators \citep{lai2009,%
xin2013} and overhead cranes \citep{rapp2012,sun2013} to aerospace vehicles \citep{consolini2010,tsiotras2000} 
and watercrafts \citep{do2009,serrano2014}. 

Basically, underactuation could arise due to the following main reasons \citep{seifried2013a}:

\begin{itemize}
\item Design issues, as for instance in the case of ships, overhead cranes, and helicopters;
\item Non-rigid body dynamics, for example if one or more flexible links are considered within a robotic manipulator;
\item Actuator failure, as is the case with aerial and underwater vehicles.
\end{itemize}

Despite this broad spectrum of applications, the problem of designing accurate controllers for underactuated systems 
is unfortunately much more tricky than for fully actuated ones. Moreover, the dynamic behavior of an UMS is frequently 
uncertain and highly nonlinear, which in fact makes the design of control schemes for such systems a challenge for 
conventional and well established methodologies.

Therefore, much effort has been made in order to improve both set-point regulation and trajectory tracking of underactuated 
mechanical systems. The most common strategies are feedback linearization \citep{seifried2013a,spong1994}, feedforward control 
by model inversion \citep{seifried2012a,seifried2012b}, adaptive approaches \citep{pucci2015,nguyen2015} sliding mode control 
\citep{ashrafiuon2008,xu2008,sankaranarayanan2009,qian2009,muske2010}, backstepping \citep{chen2012,xu2013,rudra2014}, 
controlled Lagrangians \citep{bloch2000,bloch2001} and passivity-based methods \citep{ortega2002,gomez-estern2004,ryalat2016}. 
However, it should be highlighted that the control of uncertain UMS remains hard to be accomplished, specially if a high 
level of uncertainty is involved \citep{liu2013}.

Intelligent control, on the other hand, has proven to be a very attractive approach to cope with uncertain nonlinear systems \citep{tese,rsba2010,cobem2005,nd2012,Bessa2014,Bessa2017,Bessa2018,Bessa2019,Deodato2019,Lima2018,Lima2020,Lima2021,Tanaka2013}.
By combining nonlinear control techniques, such as feedback linearization
or sliding modes, with adaptive intelligent algorithms, for example fuzzy logic or artificial neural networks, 
the resulting intelligent control strategies can deal with the nonlinear characteristics as well as with modeling 
imprecisions and external disturbances that can arise.

Due to the adaptive capabilities of artificial neural networks (ANN), it has been largely employed in the last decades to both 
control and identification of dynamical systems. In spite of the simplicity of this heuristic approach, in some situations 
a more rigorous mathematical treatment of the problem is required. Recently, much effort has been made to combine neural 
networks with nonlinear control methodology. 

On this basis, sliding mode control is an appealing technique because of its robustness against both structured and 
unstructured uncertainties as well as external disturbances. Nevertheless, the discontinuities in the control law must be 
smoothed out to avoid the undesirable chattering effects. The adoption of properly designed boundary layers have proven 
effective in completely eliminating chattering, however, leading to an inferior tracking performance, as demonstrated by 
\citet{ijac2009}. 

In this context, considering that artificial neural networks can perform universal approximation \citep{hornik1,park2}, 
\citet{ideal2012} showed that these intelligent algorithms can be successfully applied for uncertainty and disturbance 
compensation within the boundary layer of smooth sliding mode controllers. 

In this work, a sliding mode controller with a neural network compensation scheme is proposed for uncertain underactuated 
mechanical systems. On this basis, a smooth sliding mode controller is considered to confer robustness against modeling 
imprecisions and a Radial Basis Function (RBF) neural network is embedded in the boundary layer to cope with unmodeled
dynamical effects. Numerical simulations are carried out in order to demonstrate the control system performance.

\section*{INTELLIGENT CONTROL USING ARTIFICIAL NEURAL NETWORKS}

The equations of motion of a mechanical system with $n$ degrees of freedom (DOF) and $m$ actuator inputs are usually 
expressed in the following vector form \citep{seifried2013a}:

\begin{equation}
\mathbf{M(q)\ddot{q}+k(q,\dot{q})=g(q,\dot{q})+B(q)u}\:,
\label{eq:mov1}
\end{equation}

\noindent
where $\mathbf{q}\in\mathbb{R}^n$ is the vector of generalized coordinates, $\mathbf{u}\in\mathbb{R}^m$ the actuator input 
vector, $\mathbf{M(q)}\in\mathbb{R}^{n\times n}$ the positive definite and symmetric inertia matrix, $\mathbf{k(q,\dot{q})}
\in\mathbb{R}^n$ takes the Coriolis and centrifugal effects into account, $\mathbf{g(q,\dot{q})}\in\mathbb{R}^n$ represents 
the generalized applied forces, and $\mathbf{B(q)}\in\mathbb{R}^{n\times m}$ is the input matrix.

\begin{remark}
The mechanical system described in Eq.~(\ref{eq:mov1}) is called fully-actuated if $m=\mathrm{rank}(B)=n$, or underactuated 
if $m<n$.
\label{rem:rank}
\end{remark}

Considering an UMS, the vector of generalized coordinates can be partitioned as $\mathbf{q}=[\mathbf{q}_a^\mathrm{T}\mbox{ }
\mathbf{q}_u^\mathrm{T}]^\mathrm{T}$, where $\mathbf{q}_a\in\mathbb{R}^n$ and $\mathbf{q}_u\in\mathbb{R}^{n-m}$ denote, 
respectively, actuated and unactuated coordinates. In these cases, the input matrix is also conveniently assumed to be 
$\mathbf{B(q)}=[\mathbf{B}_a\mbox{ }\mathbf{B}_u]^\mathrm{T}=[\mathbf{I}\mbox{ }\mathbf{0}]^\mathrm{T}$, where $\mathbf{I}
\in\mathbb{R}^{n\times n}$ is the identity matrix. 

Therefore, for control purposes, Eq.~(\ref{eq:mov1}) may be rewritten as \citep{ashrafiuon2008,seifried2013b}:

\begin{equation}
\left[ \begin{array}{cc}
\mathbf{M}_{aa}			&\mathbf{M}_{au}\\
\mathbf{M}_{au}^\mathrm{T}	&\mathbf{M}_{uu} \end{array} \right]
\left[ \begin{array}{c} \mathbf{\ddot{q}}_a\\ \mathbf{\ddot{q}}_u \end{array} \right]=
\left[ \begin{array}{c} \mathbf{f}_a+\mathbf{u}\\ \mathbf{f}_u \end{array} \right]\:,
\label{eq:mov2}
\end{equation}

\noindent
where $\mathbf{f}_a=\mathbf{g}_a-\mathbf{k}_a$ and $\mathbf{f}_u=\mathbf{g}_u-\mathbf{k}_u$.

As described in \citep{ashrafiuon2008}, Eq.~(\ref{eq:mov2}) can be solved for the accelerations:

\begin{equation}
\mathbf{\ddot{q}}_a=\mathbf{M}_{aa}^{'-1}(\mathbf{f}_a^{'}+\mathbf{u})\:,
\label{eq:accel1}
\end{equation}
\begin{equation}
\mathbf{\ddot{q}}_u=\mathbf{M}_{uu}^{'-1}(\mathbf{f}_u^{'}-\mathbf{M}_{au}^\mathrm{T}\mathbf{M}_{aa}^{-1}\mathbf{u})\:,
\label{eq:accel2}
\end{equation}

\noindent
where $\mathbf{M}_{aa}^{'}=\mathbf{M}_{aa}-\mathbf{M}_{au}\mathbf{M}_{uu}^{-1}\mathbf{M}_{au}^\mathrm{T}$,
$\mathbf{M}_{uu}^{'}=\mathbf{M}_{uu}-\mathbf{M}_{au}^\mathrm{T}\mathbf{M}_{aa}^{-1}\mathbf{M}_{au}$,
$\mathbf{f}_a^{'}=\mathbf{f}_a-\mathbf{M}_{au}\mathbf{M}_{uu}^{-1}\mathbf{f}_u$, and
$\mathbf{f}_u^{'}=\mathbf{f}_u-\mathbf{M}_{au}^\mathrm{T}\mathbf{M}_{aa}^{-1}\mathbf{f}_a$.\vspace*{6pt}

The proposed control problem has to ensure that, even in the presence of external disturbances and modeling imprecisions, 
the vector of generalized coordinates $\mathbf{q}$ will follow a desired trajectory $\mathbf{q}^d$ in the state space. 
Hence, defining $\mathbf{\tilde{q}}=\mathbf{q}-\mathbf{q}^d$ as the tracking error vector, both trajectory tracking and
set-point regulation could be stated as $\mathbf{\tilde{q}\to0}$ as $t\to\infty$.

Consider, as for instance, the sliding mode approach, and let $m$ sliding surfaces be defined in the state space by 
$\mathbf{s(\tilde{q})=0}$, with $\mathbf{s}\in\mathbb{R}^m$ satisfying

\begin{equation}
\mathbf{s(\tilde{q})}=\alpha_a\mathbf{\dot{\tilde{q}}}_a+\lambda_a\mathbf{\tilde{q}}_a+\alpha_u\mathbf{\dot{\tilde{q}}}_u
+\lambda_u\mathbf{\tilde{q}}_u=\alpha_a\mathbf{\dot{q}}_a+\alpha_u\mathbf{\dot{q}}_u+\mathbf{s}_r\:,
\label{eq:surf1}
\end{equation}

\noindent 
where $\mathbf{s}_r=-\alpha_a\mathbf{\dot{q}}_a^d+\alpha_u\mathbf{\dot{q}}_u^d+\lambda_a\mathbf{\tilde{q}}_a+\lambda_u
\mathbf{\tilde{q}}_u$.\vspace*{6pt}

Thus, the sliding mode controller must satisfy the following Lyapunov candidate function

\begin{equation}
V(\mathbf{q,\dot{q}})=\frac{1}{2}\mathbf{s^\mathrm{T}s}
\label{eq:lyap1}
\end{equation}

\noindent
and the so called sliding condition $\dot{V}(\mathbf{q,\dot{q}})\le0$.\vspace*{6pt}

On this basis, the control law is defined as \citep{ashrafiuon2008}:

\begin{equation}
\mathbf{u}=-\mathbf{\hat{M}}_s^{-1}\left[\mathbf{\hat{f}}_s^{-1}+\mathbf{\dot{s}}_r+\kappa\sgn(\mathbf{s})\right]\:,
\label{eq:smc}
\end{equation}

\noindent
where $\mathbf{\hat{M}}_s$ and $\mathbf{\hat{f}}_s$ are, respectively, estimates of $\mathbf{M}_s=\alpha_a\mathbf{M}_{aa}^{'-1}
-\alpha_u\mathbf{M}_{uu}^{'-1}\mathbf{M}_{au}^\mathrm{T}\mathbf{M}_{aa}^{-1}$ and $\mathbf{f}_s=\alpha_a\mathbf{M}_{aa}^{'-1}
\mathbf{f}_a^{'}-\alpha_u\mathbf{M}_{uu}^{'-1}\mathbf{f}_u^{'\mathrm{T}}$. The control gain $\kappa$ should be
designed in order to take unmodeled dynamics, parametric uncertainties and external disturbances into account. \vspace*{6pt}

Therefore, regarding the development of the control law, the following assumptions should be made:

\begin{assumption}
The vector $\mathbf{f}_s$ is unknown but bounded, i.e. $|\mathbf{\hat{f}}_s-\mathbf{f}_s|=|\mathbf{d}_s|\le \mathbf{F}_s$.
\label{as:limf}
\end{assumption}
\begin{assumption}
The vector of generalized coordinates $\mathbf{q}$ is available.
\label{as:coord}
\end{assumption}
\begin{assumption}
The desired trajectory $\mathbf{q}^d$ is once differentiable with respect to time. Furthermore, every component of 
$\mathbf{q}^d$ is available and with known bounds.
\label{as:traj}
\end{assumption}

However, the presence of a discontinuous term, $\kappa\sgn(\mathbf{s})$, in the control law leads to the well 
known chattering effect. To avoid these undesirable high-frequency oscillations of the controlled variable, a thin boundary 
layer in the neighborhood of the switching surface could be defined by replacing the sign function with a smooth approximation. 
This substitution can minimize or, when desired, even completely eliminate chattering, but turns {\em perfect tracking} into 
a {\em tracking with guaranteed precision} problem, which actually means that a steady-state error will always remain 
\citep{ijac2009}. In order to improve the tracking performance, some adaptive intelligent algorithm could be used within 
the boundary layer of smooth sliding mode controllers.

At this point, since artificial neural networks can be considered as an universal approximator \citep{hornik1,park2}, we
propose the adoption of an ANN based compensation term, $\mathbf{\hat{d}}$, within the smoothed version of the control 
law presented in Eq.~(\ref{eq:smc}):

\begin{equation}
\mathbf{u}=-\mathbf{\hat{M}}_s^{-1}\big[\mathbf{\hat{f}}_s+\mathbf{\hat{d}}+\mathbf{\dot{s}}_r
+\mathbf{\kappa}\sat(\mathbf{\phi}^{-1}\mathbf{s})\big],
\label{eq:issmc}
\end{equation}

\noindent
where $\mathbf{\phi}\in\mathbb{R}^{m\times m}$ is a diagonal matrix with $m$ positive entries $\phi_i$, and $\sat(\cdot)$
is the saturation function:

\begin{equation}
\mbox{sat}(s_i/\phi_i) = \left\{\begin{array}{cl}
\mbox{sgn}(s_i)&\mbox{if}\quad |s_i/\phi_i|\ge1, \\
s_i/\phi_i&\mbox{if}\quad |s_i/\phi_i|<1\:. 
\end{array}\right.
\label{eq:sat}
\end{equation}

Due to its simplicity and fast convergence feature, radial basis functions (RBF) are used as activation functions and the related 
tracking error as input. In this case, the output of the network is defined as: 

\begin{equation}
\hat{d}(\mathbf{\tilde{x}})=\sum^M_{i=1}w_i\cdot\varphi_i(\Vert\mathbf{\tilde{x}-t}\Vert)
\label{eq:rbf}
\end{equation}

\noindent
where $\varphi_i(\cdot)$ are the activation functions and $\mathbf{t}$ a vector containing the coordinates of the center of each 
activation function.

Now, the switching vector $\mathbf{s}$ is used to train the neural network and the weights of the output layer are adjusted using 
the pseudo-inverse matrix. 

Considering a training set $T=\{(\mathbf{\tilde{x}},d)_1,(\mathbf{\tilde{x}},d)_2,\ldots,(\mathbf{\tilde{x}},d)_p\}$ and

\begin{equation}
\left[\begin{array}{cccc}
\varphi_{11} &\varphi_{12} &\cdots &\varphi_{1M}\\
\varphi_{21} &\varphi_{22} &\cdots &\varphi_{2M}\\
\vdots       &\vdots       &\ddots &\vdots      \\
\varphi_{p2} &\varphi_{p2} &\cdots &\varphi_{pM}
\end{array} \right]
\left[\begin{array}{c}
w_1\\ w_2\\ \vdots\\ w_M
\end{array} \right]
=\left[\begin{array}{c}
d_1\\ d_2\\ \vdots\\ d_p
\end{array} \right]
\quad\therefore\quad
[\varphi]\{w\}=\{d\}
\label{eq:train}
\end{equation}

\noindent
the RBF weights are computed with the pseudo-inverse $[\varphi]^+$

\begin{equation}
\{w\}=[\varphi]^+\{d\}
\label{eq:weights}
\end{equation}

\noindent
and approximation error, $E$, by the euclidean norm

\begin{equation}
E=\big\Vert\{d\}-[\varphi]\{w\}\big\Vert.
\label{eq:error}
\end{equation}

\begin{remark}
Considering that radial basis functions can perform universal approximation \citep{park2}, the output of the RBF network
can approximate $\mathbf{d}_s$ to an arbitrary degree of accuracy $\mathbf{\delta}=\mathbf{\hat{d}}^*-\mathbf{d}_s$,
where $\mathbf{\hat{d}}^*$ is the output related to set of optimal weights.
\label{re:fuzzy}
\end{remark}

Finally, Theorem~\ref{th:lyap2} shows that the smooth inteligent controller, Eq.~(\ref{eq:issmc}), ensures the convergence 
of the tracking error to the invariant set defined by the boundary layers.

\begin{theorem}
\label{th:lyap2}
Consider the uncertain underactuated mechanical system (\ref{eq:mov2}) subject to Assumptions \ref{as:limf}--\ref{as:traj}. 
Then, the controller defined by (\ref{eq:issmc}) and (\ref{eq:rbf}) ensures the convergence of the tracking errors to the 
manifold $\mathbf{\varPhi}=\{\mathbf{\tilde{q}}\in\mathbb{R}^n\:\big|\:|s_i|\le\phi_i, i=1,\ldots,m\}$.
\end{theorem}

\begin{proof}

Let a positive-definite Lyapunov function candidate $V_2$ be defined as

\begin{equation}
V(t)=\frac{1}{2}\mathbf{s}_\phi^\top\mathbf{s}_\phi\:,
\label{eq:lyap2}
\end{equation}

\noindent
where each component of $\mathbf{s}_\phi(\mathbf{\tilde{q}})$ is a measure of the distance between $s_i$ and its related 
boundary layer, and is computed as follows:

\begin{equation}
\mathbf{s}_\phi(\mathbf{\tilde{q}})=\mathbf{s}-\mathbf{\phi}\sat(\mathbf{\phi}^{-1}\mathbf{s})\:.
\label{eq:sphi}
\end{equation}

\noindent
Noting that $\mathbf{\dot{s}}_\phi=\mathbf{s}_\phi=\mathbf{0}$ inside $\mathbf{\varPhi}$, and $\mathbf{\dot{s}}_\phi=\mathbf{\dot{s}}$
outside of it, then the time derivative of $V_2$ becomes:

\begin{equation}
\begin{split}
\dot{V}(t)&=\mathbf{s}_\phi^\top\mathbf{\dot{s}}
=\mathbf{s}_\phi^\top\big(\mathbf{\alpha}_a\mathbf{\ddot{q}}_a+\mathbf{\alpha}_u\mathbf{\ddot{q}}_u+\mathbf{\dot{s}}_r\big)\\
&=\mathbf{s}_\phi^\top\big[\mathbf{\alpha}_a\mathbf{M}_{aa}^{'-1}(\mathbf{f}_a^{'}+\mathbf{u}+\mathbf{d}_a^{'})\\
&\quad+\mathbf{\alpha}_u\mathbf{M}_{uu}^{'-1}(\mathbf{f}_u^{'}-\mathbf{M}_{au}^\top\mathbf{M}_{aa}^{-1}\mathbf{u}+\mathbf{d}_u^{'})
+\mathbf{\dot{s}}_r\big]\\
&=\mathbf{s}_\phi^\top\big[\mathbf{f}_s+\mathbf{\dot{s}}_r+\mathbf{M}_s\mathbf{u}\big].
\end{split}
\label{eq:lyap2da}
\end{equation}

\noindent
Applying the control law (\ref{eq:issmc}) to (\ref{eq:lyap2da}), and noting that $\sat(\mathbf{\phi}^{-1}\mathbf{s})=
\sgn(\mathbf{s}_\phi)$ outside $\mathbf{\varPhi}$ and $|\hat{d}_i-\hat{d}_i^*|\le|\hat{d}_i|$, one obtains

\begin{equation}
\dot{V}(t)=-\mathbf{s}_\phi^\top\big[\mathbf{\hat{d}}-\mathbf{\hat{d}}^*+\mathbf{\delta}+\mathbf{\kappa}\sgn(\mathbf{s}_\phi)\big]
\le-\eta\,\Vert\mathbf{s}_\phi\Vert_1\:,
\label{eq:lyap2db}
\end{equation}

\noindent
which implies that $\mathbf{s}_\phi\to\mathbf{0}$ and $\mathbf{\tilde{q}}\to\mathbf{\varPhi}$ as 
$t\to\infty$.
\end{proof}

\section*{ILLUSTRATIVE EXAMPLE: STABILIZING THE CART-POLE UNDERACTUATED SYSTEM}

The cart-pole system is composed by a small car with an inverted pendulum on it (see Fig.~\ref{fig:cart-pole}), and the 
related equations of motion are presented as

\begin{equation}
\left[ \begin{array}{cc}
m_c+m		&ml\cos\theta\\
ml\cos\theta	&ml^2 \end{array} \right]
\left[ \begin{array}{c} \ddot{x}\\ \ddot{\theta} \end{array} \right]=
\left[ \begin{array}{c} ml\dot{\theta}^2\sin\theta\\ mgl\sin\theta \end{array} \right]+
\left[ \begin{array}{c} u\\ 0 \end{array} \right].
\label{eq:cart-pole}
\end{equation}

Here, $x$ and $\theta$ are, respectively, the position of the cart and the angular displacement of the pendulum, $m_c$ is the 
mass of the cart, $m$ and $l$ represent the concentrated mass and the length of the pendulum, and $g$ is the acceleration due 
to gravity. 

\begin{figure}[htb]
\centering
\begin{tikzpicture}[scale=1.1]
\path [shade] (-1.8,0) rectangle (3.2,-0.5);
\draw [ultra thick] (0,0) rectangle (1.4,0.8);
\draw [fill] (1.8,3.5) circle [radius=0.3];
\draw [ultra thick] (0.7,0.8) -- (1.8,3.5);
\draw [dashed] (0.7,0.8) -- (0.7,3.5);
\draw [->,thin] (0.7,2.3) to [out=0,in=157] (1.25,2.2);
\draw [->,thin] (1.5,0.4) to (2.5,0.4);
\draw [->,very thick] (-1.2,0.4) to (-0.1,0.4);
\draw [->,very thick] (-0.6,3.5) to (-0.6,2.5);
\node at (-1.4,0.4) {$u$};
\node at (0.7,0.4) {$m_c$};
\node at (2.7,0.4) {$x$};
\node at (1.05,2.5) {$\theta$};
\node at (1.3,1.8) {$l$};
\node at (1.3,1.8) {$l$};
\node at (-0.6,2.3) {$g$};
\node [white] at (1.8,3.5) {$m$};
\end{tikzpicture}
\caption{Cart-pole system.}
\label{fig:cart-pole}
\end{figure}
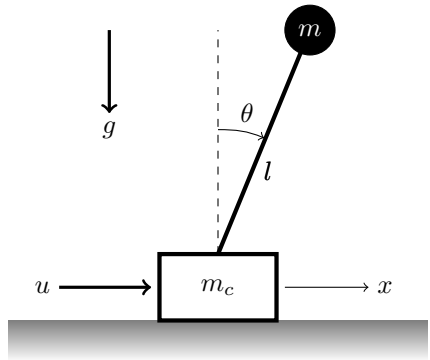

Due to Coulomb friction, a dead-zone is also considered in the cart-pole model:

\begin{equation}
u = \left\{\begin{array}{cl}
\nu+0.01&\mbox{if}\quad \nu\le0.01, \\
0&\mbox{if}\quad -0.01<\nu<0.01, \\
\nu-0.01&\mbox{if}\quad \nu\ge0.01 
\end{array}\right.
\label{eq:dzone}
\end{equation}

\noindent
where $\nu$ is the control action determined by the control law and $u$ is the actual control input.

It should be highlighted that, since only the position of the cart $x$ can be directly controlled, the angular displacement 
$\theta$ is considered an unactuated variable. On this basis, \citet{ashrafiuon2008} defined the switching variable as a 
linear combination of both actuated and unactuated state errors, $s = \alpha_a\dot{\tilde{x}} + \lambda_a\tilde{x} + \alpha_u 
l\dot{\tilde{\theta}} + \lambda_u l\tilde{\theta}$, with $\tilde{x}=x-x_d$ and $\tilde{\theta}=\theta-\theta_d$ as well as 
their derivatives representing the state errors. But considering that large uncertainties and external disturbances may 
occur, we suggest the inclusion of a compensation term $\hat{d}$ in the control law proposed in~\citep{ashrafiuon2008}, in 
order to compensate for modeling inaccuracies and to improve the control performance:

\begin{equation}
\nu=-\frac{(m_c+m\sin^2\theta)l}{\alpha_a l-\alpha_u\cos\theta}\left[\frac{(\alpha_a l-\alpha_u\cos\theta)ml\dot{\theta}^2\sin\theta
-[\alpha_a lm\cos\theta-\alpha_u(mc+m)]g\sin\theta}{(m_c+m\sin^2\theta)l}+\dot{s}_r+\hat{d}+\kappa\sat\left(\frac{s}{\phi}\right)\right]
\label{eq:control}
\end{equation}

\noindent
where $\kappa$ is the control gain, $\phi$ defines the width of the boundary layer, and $\dot{s}_r=-\alpha_a\ddot{x}_d-\alpha_u 
l\ddot{\theta}_d+\lambda_a\dot{\tilde{x}}+\lambda_u l\dot{\tilde{\theta}}$.

Considering the cart-pole system, Eq.~(\ref{eq:cart-pole}), with a dead-zone input, Eq.~(\ref{eq:dzone}), and the intelligent 
controller presented in Eq.~(\ref{eq:control}), numerical simulations were carried out to evaluate the efficacy of the proposed 
scheme. The simulation studies were performed with sampling rates of 200\,Hz for control system and 1\,kHz for dynamic model. 
The differential equations of the dynamic model were numerically solved with the fourth order Runge-Kutta method. 

In order to demonstrate the robustness of the intelligent controller against modeling imprecisions, the dead-zone input is not
considered in the development of the control law, and uncertainties of 30\% over the cart and pendulum masses are taken into 
account. On this basis, the following control parameters are considered: $m_c=0.4$\,kg, $m=0.14$\,kg, $l=0.215$\,m, 
$\alpha_a=0.02$, $\alpha_u=1$, $\lambda_a=0.005$, and $\lambda_u=2.5$.

Now, three different numerical simulations for the simultaneous stabilization of the cart, as well as the inverted pendulum at 
its unstable equilibrium point, are presented. The system is considered to be at rest but with an initial displacement in the 
pole angle: $\theta=-40^\circ$. 

First, the conventional sliding mode controller, $\hat{d}=0$, with a discontinuous term, $\kappa\sgn(s)$, is used to stabilize 
the cart-pole system. The obtained results are presented in Figs.~\ref{fig:sim1a} and~\ref{fig:sim1b}.

As observed in Fig.~\ref{fig:sim1a}, even with parametric uncertainties and an unknown dead-zone input, the conventional sliding
mode controller is able to stabilize both the cart and the pendulum in its unstable vertical position. However, Fig.~\ref{fig:sim1b}
shows that the adoption of a discontinuous term in the control law leads to the undesirable chattering effect. These high-frequency
oscillations in the controlled variable must be avoided in mechanical systems, since it can excite unmodeled vibration modes.

Therefore, retaining the compensation term as zero, $\hat{d}=0$, the discontinuous term is now substituted by a smooth interpolation 
inside the boundary layer: $\kappa\sat(s/\phi)$. Figures.~\ref{fig:sim2a}, \ref{fig:sim2b}, \ref{fig:sim2c} and~\ref{fig:sim2d}
show the obtained results.

Figure~\ref{fig:sim2b} shows that the introduction of a boundary layer can completely eliminate chattering, but, unfortunately, 
a perfect stabilization is no more possible, as observed in Fig.~\ref{fig:sim2a}. The adoption of a saturation function leads 
to limit cycles in the state variables, as shown in Figs.~\ref{fig:sim2c} and~\ref{fig:sim2d}.

In order to improve the control performance, the neural network based compensation is finally taken into account in the last
simulation study. In the first ten seconds, the collected data is used to train the neural network. Thereafter, using the 
weights computed in the training phase, the RBF network, Eq.~(\ref{eq:rbf}), is turned on and starts to compensate for the
identified modeling imprecisions. Figures.~\ref{fig:sim3a}, \ref{fig:sim3b}, \ref{fig:sim3c} and~\ref{fig:sim3d} show the 
obtained results.

As observed in Figs.~\ref{fig:sim3a}, \ref{fig:sim3c} and~\ref{fig:sim3d}, the control performance is clearly enhanced. 
Figure~\ref{fig:sim3c}, for example, shows that the limit cycle almost disappears with the adoption of the neural
network based compensation scheme. As a matter of fact, the improved performance of the proposed scheme is due to its 
ability to recognize and compensate for modelling inaccuracies. This enhanced performance is possible even with the
introduction of a thin boundary layer, which can successfully eliminate chattering, as shown in Fig.~\ref{fig:sim3b}. 

\begin{figure}[htb]
\centering
\includegraphics[width=\textwidth]{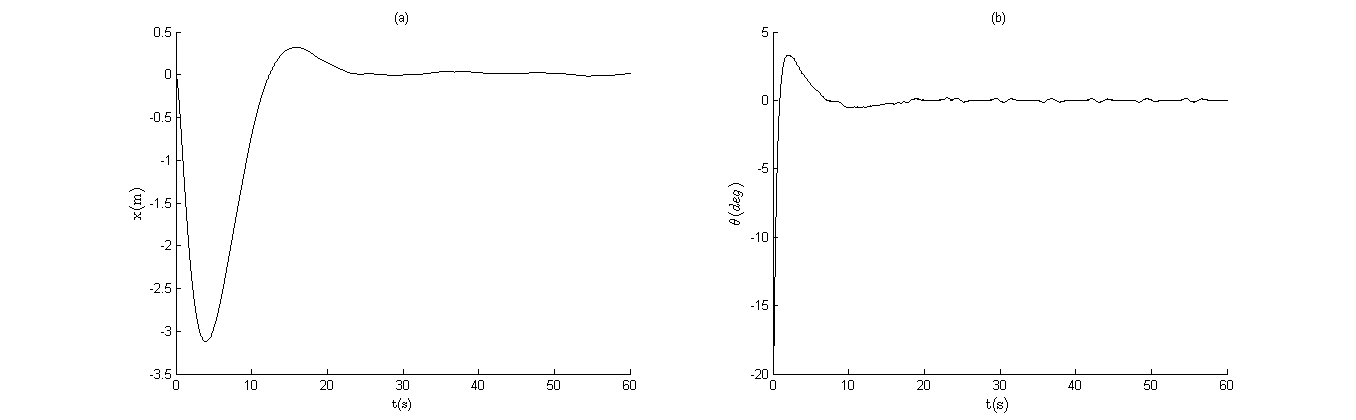}
\caption{Cart position (a) and pole angle (b) with the discontinuous sliding mode controller.}
\label{fig:sim1a}
\end{figure}

\begin{figure}[htb]
\centering
\includegraphics[width=0.5\textwidth]{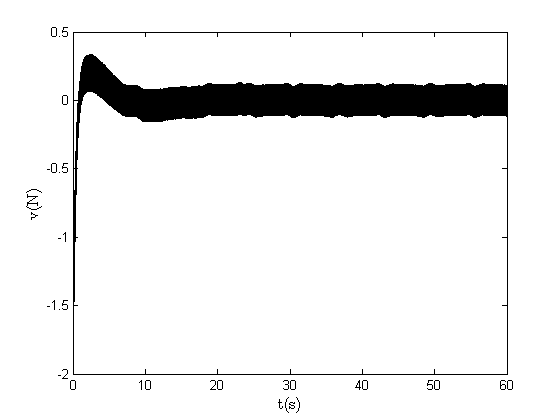}
\caption{Control action for the discontinuous sliding mode controller.}
\label{fig:sim1b}
\end{figure}

\begin{figure}[htb]
\centering
\includegraphics[width=\textwidth]{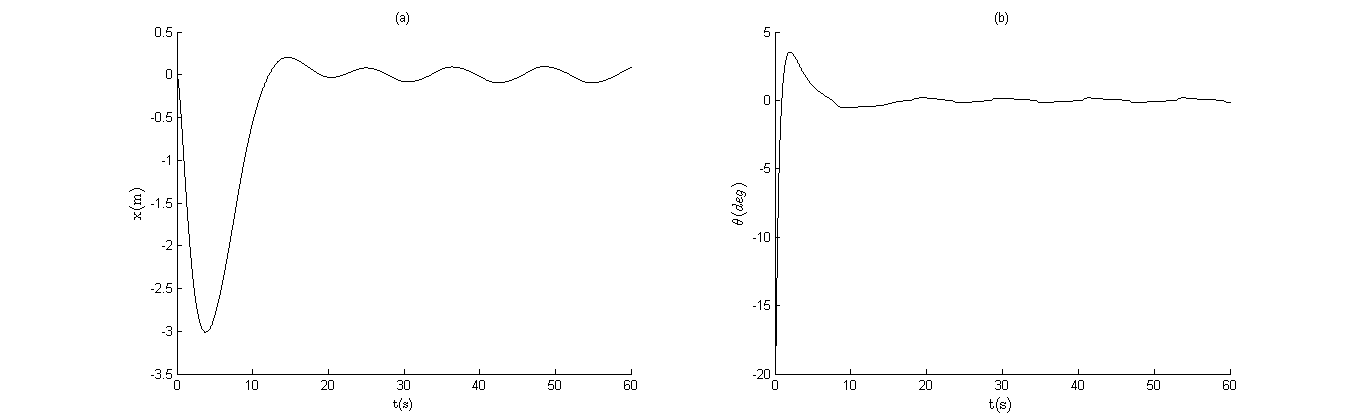}
\caption{Cart position (a) and pole angle (b) with the smoothed version of the sliding mode controller.}
\label{fig:sim2a}
\end{figure}

\begin{figure}[htb]
\centering
\includegraphics[width=0.47\textwidth]{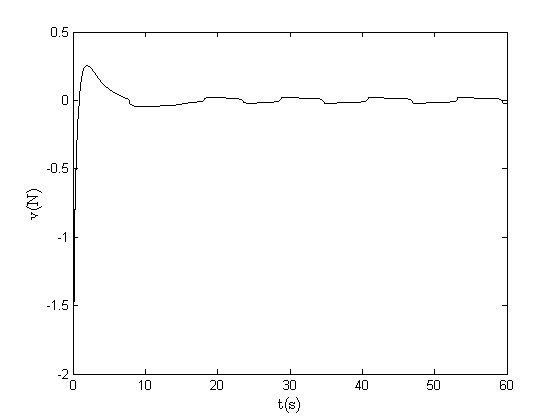}
\caption{Control action for the smoothed version of the sliding mode controller.}
\label{fig:sim2b}
\end{figure}

\begin{figure}[htb]
\begin{minipage}{0.5\textwidth}
\centering
\includegraphics[width=0.85\textwidth]{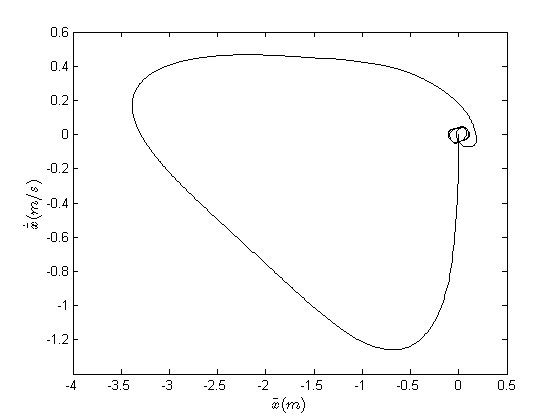}
\caption{Phase portrait for the cart variables with the smoothed version of the sliding mode controller.}
\label{fig:sim2c}
\end{minipage}
\begin{minipage}{0.5\textwidth}
\centering
\includegraphics[width=0.85\textwidth]{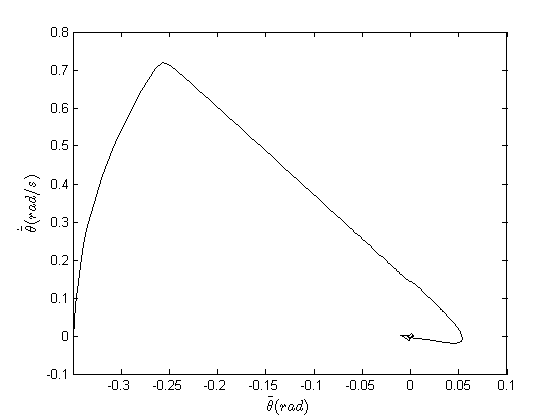}
\caption{Phase portrait for the pole variables with the smoothed version of the sliding mode controller.}
\label{fig:sim2d}
\end{minipage}\end{figure}

\begin{figure}[htb]
\centering
\includegraphics[width=\textwidth]{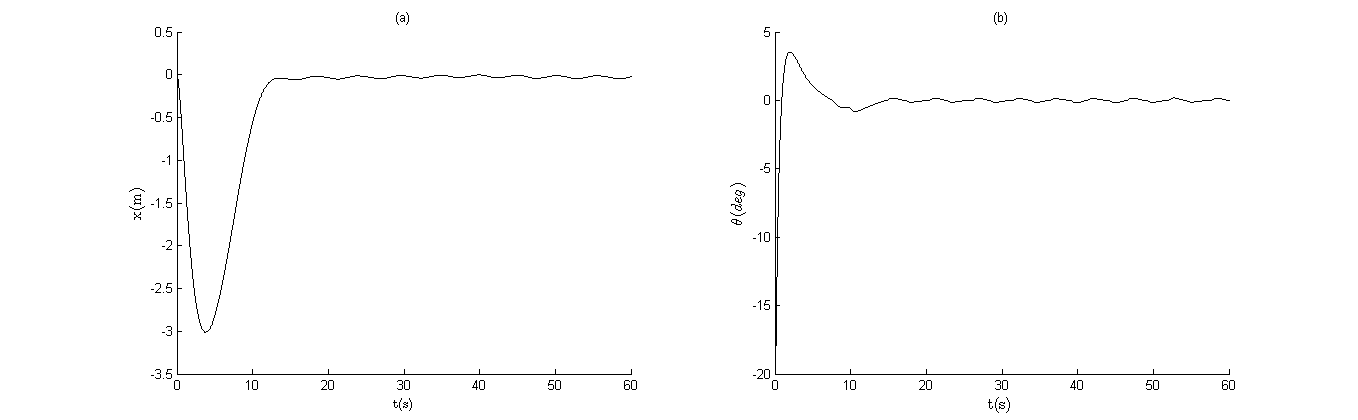}
\caption{Cart position (a) and pole angle (b) with the intelligent sliding mode controller.}
\label{fig:sim3a}
\end{figure}

\begin{figure}[htb]
\centering
\includegraphics[width=0.47\textwidth]{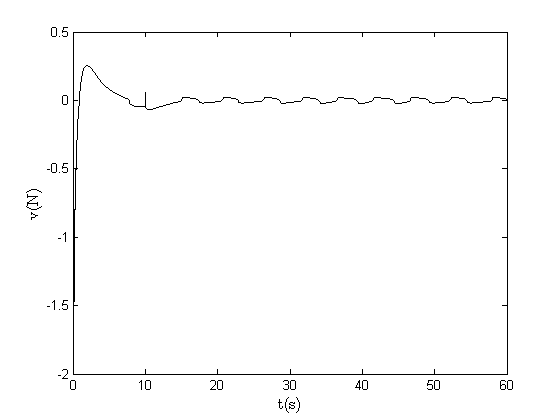}
\caption{Control action for the intelligent sliding mode controller.}
\label{fig:sim3b}
\end{figure}

\begin{figure}[htb]
\begin{minipage}{0.5\textwidth}
\centering
\includegraphics[width=0.85\textwidth]{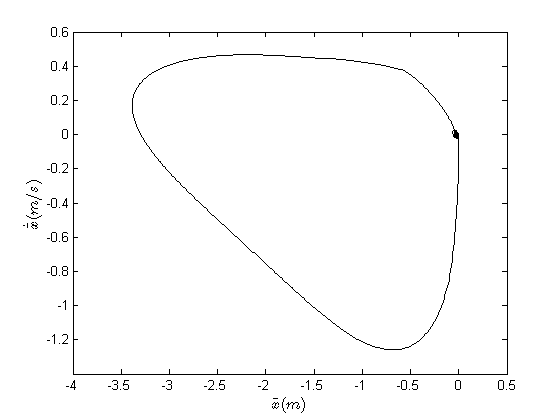}
\caption{Phase portrait for the cart variables with the intelligent sliding mode controller.}
\label{fig:sim3c}
\end{minipage}
\begin{minipage}{0.5\textwidth}
\centering
\includegraphics[width=0.85\textwidth]{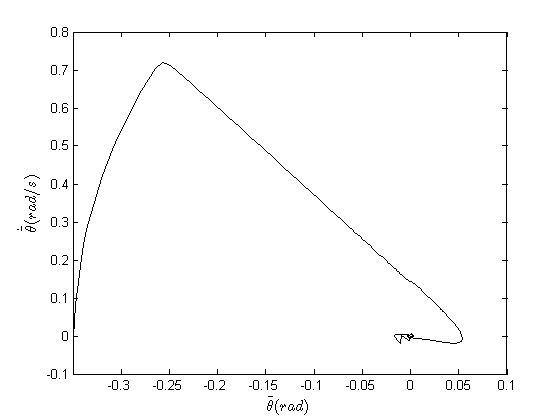}
\caption{Phase portrait for the pole variables with the intelligent sliding mode controller.}
\label{fig:sim3d}
\end{minipage}\end{figure}

\section*{CONCLUDING REMARKS}

The present work addressed the problem of controlling uncertain nonlinear underactuated systems with a sliding mode control 
approach, but enhanced with a RBF neural network. The convergence properties of the resulting intelligent controller are proved 
by means of a Lyapunov-like stability analysis. In order to illustrate the controller design method, the proposed scheme is 
applied to a cart-pole system. The control system performance is confirmed by means of numerical simulations. The adoption of 
a RBF network provides an smaller balancing error due to its ability to compensate the performance drop caused by the change 
of the sign function for the saturation function.

\section*{ACKNOWLEDGMENTS}

The authors would like to acknowledge the support of the Alexander von Humboldt Foundation, the Brazilian Coordination 
for the Improvement of Higher Education Personnel and the Brazilian National Research Council, the Brazilian National 
Agency of Petroleum, Natural Gas and Biofuels and the German Academic Exchange Service.

\end{document}